\documentclass[11pt]{article}
\usepackage[paperwidth=8.5in, paperheight=11in, top=1.25in, bottom=1.25in, left=1.00in, right=1.00in]{geometry}

\usepackage{imakeidx} 
	\makeindex
\usepackage[dvipsnames]{xcolor} 
\usepackage[plainpages=false, pdfpagelabels]{hyperref} 
 	\hypersetup{
 		colorlinks   = true,
 		citecolor    = blue,
 		linkcolor    = Maroon,
 		urlcolor     = Turquoise
 		}
\usepackage{natbib}

\usepackage{amssymb,amsmath,amsthm,mathtools} 
\mathtoolsset{showonlyrefs=true}
\usepackage{ragged2e}
\usepackage[justification=RaggedRight]{caption}
\usepackage[all,error]{onlyamsmath} 
\usepackage{empheq} 
\usepackage{graphicx}

\newcommand{\mockalph}[1]{}

\renewcommand{\P}{\mathbb{P}}

\newcommand{\E}{\mathbb{E}}

\newcommand{\R}{\mathbb{R}}

\def\ud{\mathrm{d}}

\def\ind{{\mathchoice{1\mskip-4mu\mathrm l}{1\mskip-4mu\mathrm l}
{1\mskip-4.5mu\mathrm l}{1\mskip-5mu\mathrm l}}}

\newtheorem{theorem}{Theorem}[section]
\newtheorem{definition}[theorem]{Definition}
\newtheorem{corollary}[theorem]{Corollary}
\newtheorem{proposition}[theorem]{Proposition}
\newtheorem{remark}[theorem]{Remark}

\title{When to Sell an Asset? -- A Distribution Builder Approach\footnote{This article has been written during Stephan Sturm's sabbatical leave at NYU, hosted by Peter Carr. All content was finished before Peter Carr’s untimely passing; Stephan Sturm edited the manuscript only minimally for publication and would like to thank Sixian Jin for helpful feedback. AI (GPT 5.6 Sol) was used for a final error and typo check; Stephan Sturm alone is responsible for all remaining mistakes.}}

\author{
	Peter Carr ($\dagger$)\thanks{%
		P. Carr, New York University, Tandon School of Engineering, Department of Finance and Risk Engineering, One MetroTech Center, Brooklyn, NY 11201}
	\and
	Stephan Sturm \thanks{ %
		S. Sturm, Worcester Polytechnic Institute, Department of Mathematical Sciences, 100 Institute Road, Worcester, MA 01609, USA
		(e-mail: \texttt{ssturm@wpi.edu})}
}

\begin{document}
	
\maketitle

\begin{abstract}
	We consider the question of the optimal timing of the sale of an asset with stochastic dynamics. Our analysis is based on the method of the distribution builder introduced by Sharpe, Goldstein and Blythe \cite{SGB00} for the purpose of optimal portfolio selection. Instead of specifying a utility function or risk aversion coefficient, this tool directly elicits the target distribution of the investor. We show how the problem of an optimal asset sale is in this setting linked to the problem of finding a Skorokhod embedding of a distribution into a diffusion process. In the case where the asset process follows a geometric Brownian motion and a specific family of distributions is targeted, one can observe a risk-return tradeoff. 
\end{abstract}

\vspace{5mm}

\begin{flushleft}
	\textbf{Keywords:} Asset Sale, Optimal Stopping, Skorokhod Embedding, Distribution Builder.\\
	\textbf{Mathematics Subject Classification (2020):} 60G40, 60J70, 91G10.\\
	\textbf{JEL classification:} C02, G11, D81, C61.
\end{flushleft}

\vspace{1cm}

\section{Introduction}
We consider the problem of the optimal timing of the sale of an asset. This problem has been considered in different versions by several authors. For example, Shiryaev, Xu and Zhou as well as du Toit and Peskir \cite{SXZ08, dTP09} consider a universal (i.e., preference-independent) setting and find the stopping time at which the price is closest to the ultimate maximum price on a given time horizon. Leung and Wang \cite{LW19} consider the problem of finding the stopping time that maximizes the expected utility of the asset sale.

In the following we take an approach that takes into account the seller's preferences, but avoids the use of utility functions that are hard to estimate in practice. Following the distribution builder approach of Sharpe and co-authors \cite{S07, GJS08, SGB00}, we assume that the seller has a target distribution in mind which she wants to achieve for the (discounted) value $X$  of the asset she wishes to sell. The goal is to determine if this target distribution can be actually achieved, and if so how a stopping time can be determined that leads to the desired distribution.

The paper is structured as follows. In Section \ref{sec:set} we lay out the setting and provide a precise mathematical formulation of the problem. Before considering the general problem, we investigate in detail in Section \ref{sec:GBM} the case in which the asset dynamics follow a geometric Brownian motion. Not only is this the standard set-up in the literature for optimal asset sales, but we can without much technicality show the main features of our approach. We investigate the case of particular parametrized families of target distributions (such as log-normal, Pareto and Weibull) and show that in this context the optimal choice of the target distribution amounts to considering a risk--return tradeoff similar to the efficient frontier in Markowitz portfolio optimization. Section \ref{sec:gen} provides the general analysis for the case of (non-singular) diffusion processes as models for the asset price process. Section \ref{sec:conc} concludes; the appendix provides further examples.

\section{Setting}\label{sec:set}

We consider a (discounted) asset price process described by diffusion dynamics
\begin{equation}
	\ud X_t = \mu(X_t) \,\ud t + \sigma(X_t) \, \ud W_t, \qquad X_0 =x.
\end{equation}
The asset is assumed to be liquid and indivisible. The owner of the asset is interested in selling the asset, and she wants to do so in a way that is optimal according to her preferences. The specification of the preferences is done according to the methodology of the distribution builder (see \cite{S07, GJS08, SGB00}): The seller is able to choose a distribution (identified with its cumulative distribution function, or cdf, $F$) with the distribution builder tool. The engine behind the distribution builder should provide information about whether the distribution given can actually be reached by selling the asset at some random (stopping) time, providing a strategy that can be implemented for the sale. It might additionally provide insight into whether the distribution specified is actually optimal, or if there is a way to actually do better.

Specifically: The seller provides a distribution $F$ with the goal to sell the asset at the stopping time $\tau$ such that 
\begin{equation}\label{dist}
X_\tau \sim F
\end{equation}
(where $\sim$ stands for equality in distribution). This might be generalized: in some cases it might not be possible to attain the exact distribution desired, but one might be able to achieve a distribution that is even better, i.e.,
\begin{equation}\label{dist-sup}
X_\tau \succeq F,
\end{equation}
where $\succeq$ denotes first-order stochastic dominance (i.e., $F(z) \geq \P[X_\tau \leq z ]$ for all $z \in \mathbb{R}$). This can be seen as a way to allow the seller to 'throw money away' (a random, non-negative amount), similar to the concept of superhedging generalizing the idea of hedging. 

\begin{definition}
	Let $F$ be a distribution. We say $F$ is \emph{attainable} if \eqref{dist} holds true, and we call $F$ \emph{super-attainable} if \eqref{dist-sup} does.
\end{definition}

In the following we aim to answer the following questions key to a successful understanding and implementation of the distribution-builder framework for the sales-timing problem:
\begin{enumerate}
	\item Which distributions $F$ are \textit{attainable} (resp. \textit{super-attainable}), i.e., for which there exists an a.s. finite stopping time $\tau$ such that \eqref{dist} (resp. \eqref{dist-sup}) holds true?
	\item Can the stopping time $\tau$ be characterized in a meaningful way such that the liquidation policy yielding \eqref{dist} (resp. \eqref{dist-sup}) can be readily implemented?
	\item What can we say about the distribution of the stopping time $\tau$? Which further properties does it possess?
\end{enumerate}

\section{Motivational Example: Geometric Brownian Motion}\label{sec:GBM}

To fix ideas, let us assume that the discounted asset has geometric Brownian motion dynamics

\[
\ud X_t = (\mu -r) X_t \, \ud t + \sigma X_t \, \ud W_t, \qquad X_0 = x,
\] 
where the discounting rate $r$ reflects the time preferences of the asset seller and does not correspond necessarily to a riskless interest rate -- in general we want to emphasize that the current paper involves no hedging arguments at all.

In this setting we can answer the first two questions asked above in the affirmative, and the answers rely on the Az\'{e}ma--Yor (barrier-type) solution to the Skorokhod embedding problem for Brownian motion with drift as described in Grandits and Falkner \cite{GF00}; see also Ob\l oj \cite{Obl04}. Specifically, the explicit answers are:

\begin{theorem}
	A distribution $F$ supported on $\mathbb{R}_{>0}$ is attainable if and only if it satisfies the moment condition
	\begin{equation}\label{moment}
	\int \Bigl(\frac{z}{x}\Bigr)^A \, F(\ud z) \leq 1, \qquad A = 1-\frac{2(\mu-r)}{\sigma^2}.
	\end{equation}
\end{theorem}

We say a distribution that satisfies \eqref{moment} with equality is \textit{optimal}.

\begin{proof}
	While this can be directly proved using the results of \cite[Corollaries 2.1 and 2.2]{GF00} about a Brownian motion with drift via the transformation
	\begin{equation}\label{expo}
	Z_t := \frac{\log{(X_t)} - \log{(x)}}{\sigma} = \biggl(\frac{\mu-r}{\sigma} - \frac{\sigma}{2}\biggr) t + W_t,
	\end{equation}
we just treat this as a special case of the general result given in Theorem \ref{thm:sko-gen}.
\end{proof}

With the help of this theorem, we can characterize all super-attainable distributions thereby paving the way for the economic key message.

\begin{theorem}\label{super-gbm}
	Depending on the relation of excess drift and volatility, we get the following results:
	\begin{itemize}
	\item[i)] If $\mu-r \geq \sigma^2/2$ (equivalently $A \leq 0$) then any distribution $F$ is super-attainable.
	\item[ii)] If $\mu-r < \sigma^2/2$ (equivalently $A > 0$), a distribution $F$ is super-attainable if and only if \eqref{moment} holds. In this case, moreover,
	\begin{itemize}
		\item[a)] If $ 0 < \mu -r < \sigma^2/2$ (equivalently $ 0< A < 1$ ) and $F$ is optimal, non-constant and with finite first moment $m = \int z \, F(\ud z)$, then $m > x$.
		\item[b)] If $\mu \leq r$ (equivalently $A \geq 1$), then $F$ has a finite first moment $m = \int z \, F(\ud z)$ satisfying $m \leq x$.
	\end{itemize}
	\end{itemize}
\end{theorem}

\begin{proof}
That all attainable distributions (i.e., satisfying \eqref{moment}) are super-attainable is clear.

In case i), it remains to show that also all other distributions are super-attainable. We aim for finding a distribution dominating $F$ in first-order such that the moment inequality \eqref{moment} holds with equality. Define a family of distributions $F^c$ by defining
\[
F^c(z) := \left\{ \begin{array}{ll} F(z) & \text{ if } z \geq c,\\ 0 & \text{ else. }\end{array}\right.
\]
Note that $F^c \succeq F$ by construction. Now set $\tilde{c} := \inf\{ c> 0 : \int (z/x)^A \, F^c(\ud z) \leq 1 \}$. In the case that $F$ is continuous, $F^{\tilde{c}}$ gives us already the solution. In the general case set 
\[
\tilde{x} = \int (z/x)^A \, F^{\frac{\tilde{c}}{2}}(\ud z), \qquad \tilde{y} = \int (z/x)^A \, F^{\tilde{c}}(\ud z)
\]
(where $\tilde{y} \leq 1 \leq \tilde{x}$ thanks to the right-continuity of the distribution function), and define the distribution $\tilde{F}$ by
\[
\tilde{F}(x) := \frac{1-\tilde{y}}{\tilde{x} - \tilde{y}} F^{\frac{\tilde{c}}{2}}(x) + \frac{\tilde{x}-1}{\tilde{x} - \tilde{y}} F^{\tilde{c}}(x)
\]
Now $\tilde{F} \succeq F$ (as $\tilde{F}(x) \leq F(x)$ for $x < \tilde{c}$ and $\tilde{F}(x) = F(x)$ for $x \geq \tilde{c}$), while it satisfies the moment condition \eqref{moment}.

To show ii), assume that $F$ is super-attainable, i.e., there exists some $F^*$ satisfying \eqref{moment} such that $F^* \succeq F$. However, as $\mu-r < \sigma^2/2$ implies $A>0$, the mapping $z \mapsto z^A$ is strictly increasing and therefore
\[
\int \Bigl(\frac{z}{x}\Bigr)^A \, F(\ud z) \leq \int \Bigl(\frac{z}{x}\Bigr)^A \, F^*(\ud z) \leq 1.
\]

To show the fine points in this case, we note that in ii)\ a) we have $0 < \mu -r < \sigma^2/2$ whence $A \in (0,1)$ and the power function $z \mapsto z^A$ is strictly concave on the nonnegative reals. Thus, assuming that an optimal distribution $F$ has a first moment, it follows from Jensen's inequality
\[
1 = \int \Bigl(\frac{z}{x}\Bigr)^A \, F(\ud z) < \biggl(\frac{1}{x}\int z \, F(\ud z)\biggr)^A = \Bigl(\frac{m}{x}\Bigr)^A
\]
and thus $m > x$.

Finally, when $r \geq \mu$ we have $A \geq 1$ and thus $F$ has to have a finite first moment. Again, using Jensen's inequality we note that
\[
1 \geq \int \Bigl(\frac{z}{x}\Bigr)^A \, F(\ud z) \geq \biggl(\frac{1}{x}\int z \, F(\ud z)\biggr)^A = \Bigl(\frac{m}{x}\Bigr)^A
\]
whence $m \leq x$ as claimed.
\end{proof}

From an economic perspective, this is best analyzed from the point of view of the discount rate. In the case that $\mu -r \geq \sigma^2 / 2$ the discount rate is so low that indeed all distributions are super-attainable; there is no restriction for the holder to get what she wishes for, or conversely, there seems to be no good reason to sell the asset with any specific target distribution in mind, as a strictly better outcome can also be achieved. The asset holder might, absent forcing external criteria, not be interested in selling at all. Conversely, if $r \geq \mu$ the discount rate is so high that only distributions with a mean below the current asset price are super-attainable. From basic principles of risk aversion and the preference for payoffs with a higher mean over those with a lower, it follows that the asset should be sold immediately.

The most interesting case occurs when $0 < \mu -r < \sigma^2/2$. Here optimal stopping can help to end up in a distribution with mean actually higher than the current asset price, though at the price of taking on uncertainty. This can be seen as similar to the risk-return trade-off in the efficient frontier of Markowitz's portfolio optimization. We will revisit this in the specific examples at the end of this section. For the moment we want to concentrate on the question of how optimal stopping can actually be achieved in this case.

The construction of the stopping time is based on the Az\'{e}ma--Yor \cite{AY79a, AY79b} solution to the classical Skorokhod embedding problem.   

\begin{proposition}
	In the case that $\mu - r < \frac{\sigma^2}{2}$, an optimal distribution $F$ (i.e., \eqref{moment} is satisfied with equality) can be attained by setting $\tau$ to be the first hitting time
			\begin{align*}
			\tau & =\inf\bigl\{ t > 0 \, : \, M_t \geq \Psi_F^{\mu, \sigma}(X_t)\bigr\}\\
			& =\inf\bigl\{ t > 0 \, : \, (M_t, X_t) \in \mathcal{D}_F \bigr\}, \quad \mathcal{D}_F := \bigl\{(m,x) \in \mathbb{R}^2 \, : \, m \geq \Psi_F^{\mu, \sigma}(x)\bigr\}
			\end{align*}
			where $M$ is the running maximum of $X$, $M_t = \sup_{s \leq t} X_s$, and 
			\[
			\Psi_F^{\mu, \sigma}(z) = \left\{ \begin{array}{ll} \biggl( \frac{1}{1-F(z-)} \int_{[z,+\infty)} u^A \, F(\ud u) \biggr)^\frac{1}{A}& \text{ if } F(z-) <1, \\ z & \text{ else}. \end{array}\right.
			\]
\end{proposition}

\begin{proof}
	This follows directly from the results of Grandits and Falkner \cite[Corollaries 2.1 and 2.2]{GF00} for the Skorokhod problem for drifted Brownian motion by means of an exponential transformation. Specifically, the authors show that if $\int e^{-2\kappa x} \nu(dx) = 1$ for a measure $\nu$ and $\kappa >0$, then a stopping time $\tau$  for which the drifted Brownian motion $X^\kappa_t = \kappa t + B_t$ stopped at $\tau$ has distribution $\nu$ can be represented as 
	\[
	\tau = \inf\bigl\{ t > 0 \, : \, M_t^\kappa \geq \Psi_\nu^\kappa(X_t^\kappa)\bigr\}, \qquad M_t^\kappa = \sup_{s \leq t} X_s^\kappa  
	\]
	for $\Psi_\nu^\kappa = u^{-1}\circ \psi_\nu \circ u$ with scale function $u(z) = 1- e^{2\kappa z} = 1-e^{\sigma A z}$ and barycenter function of the right tail
	\[
	\psi_\nu (y) = \left\{ \begin{array}{ll} \frac{1}{\nu([y,\infty))} \int_{[y,\infty)} z \, \nu(\ud z) & \text{ if } \nu\bigl([y,\infty)\bigr)>0 , \\ y & \text{ else}. \end{array}\right.
	\]
	Switching the sign of the Brownian motion to obtain results for negative $\kappa$ and applying again the substitution \eqref{expo} yields the result.
\end{proof}

In this case we can also give an explicit description of the expected selling time.

\begin{proposition}
	The expected time to attain an optimal distribution in the case $0<\mu - r < \frac{\sigma^2}{2}$ is finite if and only if $\int \vert \log(z) \vert \, F(\ud z) < \infty$. In that case it is given by
	\[
	\E[\tau] = \frac{2}{\sigma^2 -2(\mu-r) } \int \log\Bigl(\frac{x}{z}\Bigr) \, F(\ud z).
	\]
\end{proposition}

\begin{proof}
	\cite[Proposition 2.2]{GF00} shows that $\E[\tau] = -\frac{1}{\kappa} \E[Y_\tau]$ whenever $\int \vert z \vert  \, \nu(\ud z) < \infty$. Using again the exponential transformation \eqref{expo} gives the result.
\end{proof}

If we add the additional condition that the sale has to be done before a specified time $T$, similar results can be established using the results of Ankirchner, Hobson and Strack \cite{AHS15} to provide a sufficient condition -- at the price of losing the nice intuitive construction of the Az\'{e}ma--Yor embedding:
	
\begin{proposition}
	Denote by $\Phi$ and $\varphi$ the cumulative distribution function (cdf) and density of a standard normal distribution. Let $F$ be an optimal target distribution with density $f$ and assume that $g:= F^{-1} \circ \Phi$ is absolutely continuous. $F$ is attainable before time $T$ if
	\[
	\varphi(z) \leq \sigma \bigl(f\circ g\bigr)(z) g(z) \sqrt{T}
	\]
	on $\{f\bigl(g(z)\bigr)>0\}$.
\end{proposition}

\subsection{Specific Families of Target Distributions}

In the following we consider the case where the prospective asset seller already has a family of target distributions in mind from which she wants to choose. It turns out that in the case we are most interested in, $0< \mu - r < \sigma^2/2$,  the problem can often be cast as a mean-variance trade-off, similar to Markowitz's portfolio optimization and the efficient frontier. We are considering the following cases:

\textbf{Log-normal}: Consider the case where the investor is interested in liquidating with a log-normal target distribution, $F_{a,b} \sim \mathcal{LN}(b,a^2)$, $F_{a,b}(z) = \int_0^z \frac{1}{\sqrt{2\pi}ay} e^{-\frac{(\log(y)-b)^2}{2a^2}} \, \ud y$. Then we have
	\[
\int_0^\infty z^A \, dF_{a,b}(z) = \frac{1}{\sqrt{2\pi} a}\int_{-\infty}^\infty e^{Ay} e^{-\frac{(y-b)^2}{2a^2}} \, \ud y = e^{A\bigl(b + A\frac{a^2}{2}\bigr)}
\]
and given the initial asset value $x$, the attainability condition \eqref{moment} reduces to
\[
b \leq \log(x) - A \frac{a^2}{2}.
\]
Alternatively, in terms of mean $m = e^{b + \frac{a^2}{2}}$ and variance $s^2 = (e^{a^2}-1)m^2$ of the log-normal distribution, this gives a mean-variance trade-off
\[
m^{2(2-A)} \leq x^2 \bigl(m^2+s^2\bigr)^{1-A}.
\]
We see that if $s \to \infty$, the right-hand side diverges to infinity, so distributions of arbitrary mean are attainable (at the price of increasing variance).  For an optimal distribution, the barrier function $\Psi_{F_{a,b}}^{\mu, \sigma}$ can be calculated explicitly,
\[
\Psi_{F_{a,b}}^{\mu, \sigma}(z) =\Biggl(\frac{ \Phi \left(\frac{b-\log (z)}{a}+ a A\right)}{\Phi\Bigl(\frac{b-\log(z)}{a}\Bigr)}\Biggr)^{\frac{1}{A}}e^{b+\frac{a^2}{2} A}
\]
where $\Phi$ is the cumulative distribution function of the standard normal distribution. The expected selling time is 
\[
\E[\tau] = \frac{2\bigl(\log(x)-b\bigr)}{\sigma^2 - 2(\mu-r)}.
\]

If we make the additional restriction that we consider only log-normal distributions that can be achieved before time $T$, we can do so for optimal $F$ as long as $a^2 \leq \sigma^2 T$.

\begin{center}
	\begin{figure}[htb]
		\includegraphics[width=0.45\linewidth]{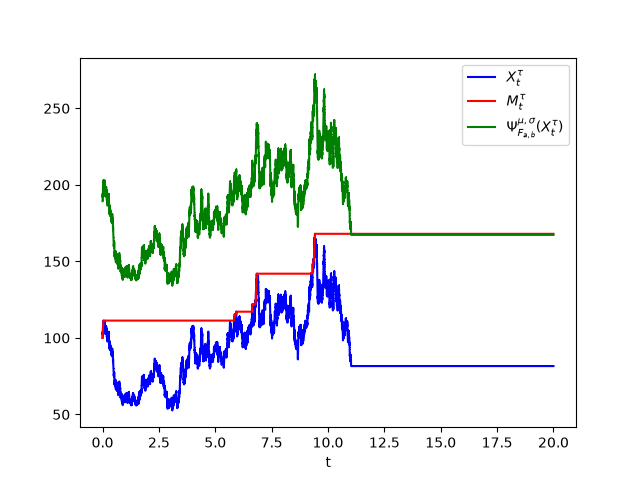} \qquad
		\includegraphics[width=0.45\linewidth]{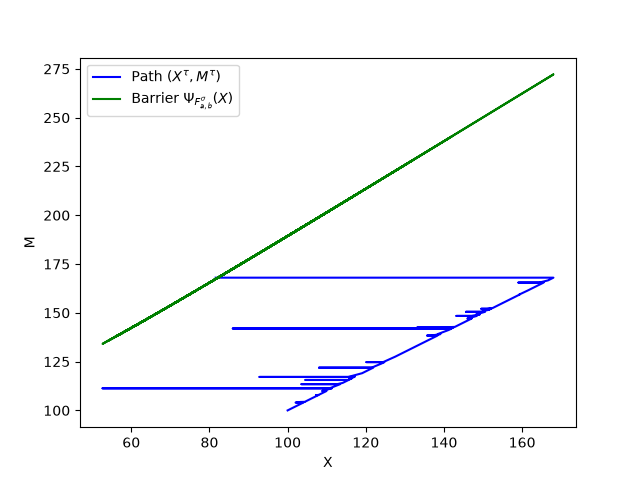}
		\caption{Az\'{e}ma--Yor barrier function $\Psi^{\mu,\sigma}_{F_{a,b}}(X_t)$ for a log-normal target distribution (parameters: $X_0 =100$, $\mu = 0.05$, $r=0.01$, $\sigma = 0.3$, $a = 0.8$). Left: Geometric Brownian motion (blue), running maximum (red) and running barrier (green) stopped upon hitting the barrier; Right: Spatial $(X_t, M_t)$ diagram with barrier function. Calculated parameters: $A \approx 0.11$, $b=4.57$, $m \approx 132.90$, $s \approx 125.84$, $\E[\tau]\approx 7.11$.}
	\end{figure}
\end{center}

\textbf{Pareto}: In the case of a Pareto distribution, $F_{x_0,p}(z) = \bigl(1-\bigl(\frac{x_0}{z}\bigr)^p\bigr)\ind_{\{z \geq x_0\}}$, we have
	\[
	\int_0^\infty z^A \, dF_{x_0,p}(z) = px_0^p \int_{x_0}^\infty z^A \frac{1}{z^{p+1}} \, \ud z = \left\{ \begin{array}{ll}\frac{px_0^A}{p-A} & \text{ if } p > A, \\ \infty & \text{ else. } \end{array}\right.
	\]
	Thus the distribution is only attainable if $p > A$. Moreover, in this case the condition \eqref{moment} on the initial asset value $x$ reduces to
	\[
	\frac{p}{p-A} \leq \Bigl(\frac{x}{x_0}\Bigr)^A.
	\]
	If, moreover, we have $p>2$, we can rewrite the condition in terms of mean $m = \frac{px_0}{p-1}$ and variance $s^2 = \frac{px_0^2}{(p-1)^2(p-2)}$ of the Pareto distribution targeted. We get
	\[
	\frac{1+\sqrt{1+m^2/s^2}}{1-A+\sqrt{1+m^2/s^2}} \leq \biggl(\frac{x(1+\sqrt{1+m^2/s^2})}{m\sqrt{1+m^2/s^2}}\biggr)^A.
	\]
	From the asymptotics $s \to \infty$ it follows that not all means are attainable, but only those satisfying $m \leq 2x((2-A)/2)^\frac{1}{A}$.
The explicit form of the barrier for optimal distributions is
\[
\Psi_{F_{x_0,p}}^{\mu,\sigma}(z) =\left\{\begin{array}{ll} x & \text{ if } z \leq x_0, \\ \frac{x}{x_0} z & \text{ if } z > x_0. \end{array}\right. 
\]
For the expected selling time this gives
\[
\E[\tau] =  \frac{2\bigl(\log(\frac{x}{x_0}) - \frac{1}{p}\bigr)}{\sigma^2 - 2(\mu-r)}.
\]

The sufficient condition for a bounded optimal embedding before time $T$ translates here into
\[
\frac{\varphi(z)}{1-\Phi(z)} \leq p \sigma \sqrt{T} 
\]
and as the left side grows asymptotically as $z$, this bound is not satisfied for any choice of the parameter $p$.

	\begin{center}
	\begin{figure}
		\includegraphics[width=0.45\linewidth]{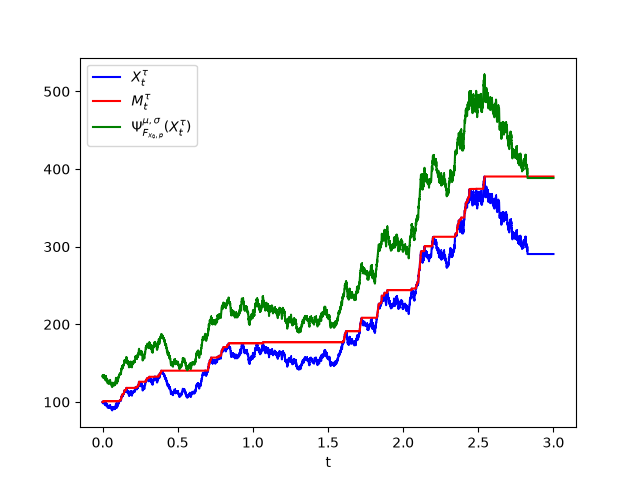} \qquad
		\includegraphics[width=0.45\linewidth]{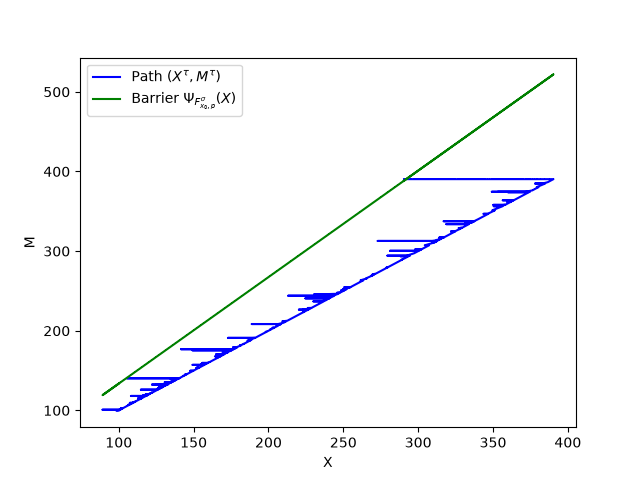}
		\caption{Az\'{e}ma--Yor barrier function $\Psi^{\mu, \sigma}_{F_{x_0,p}}(X_t)$ for a Pareto target distribution (parameters: $X_0 =100$, $\mu = 0.05$, $r=0.01$, $\sigma = 0.3$, $p=3.5$). Left: Geometric Brownian motion (blue), running maximum (red) and running barrier (green) stopped upon hitting the barrier; Right: Spatial $(X_t, M_t)$ diagram with barrier function. Calculated parameters: $A \approx 0.11$, $x_0 \approx 74.80$, $m \approx 104.72$, $s \approx 45.70$, $\E[\tau]\approx 0.93$.}
	\end{figure}
\end{center}

\textbf{Weibull}: For the Weibull distribution, $F_{k,\lambda}(z) = \bigl(1-e^{-(z/\lambda)^k}\bigr)\ind_{\{z \geq 0\}}$,
\[
\int_0^\infty z^A \, dF_{k,\lambda}(z)  = \int_{0}^\infty z^A \frac{k}{\lambda} \Bigl(\frac{z}{\lambda}\Bigr)^{k-1} e^{-(\frac{z}{\lambda})^k} \, \ud z = \lambda^{A} \int_{0}^\infty  w^\frac{A}{k} e^{-w} \, dw = \lambda^{A}\Gamma\Bigl(1+\frac{A}{k}\Bigr) 
\]
The condition on the initial asset value $x$ reads then
\[
 \Gamma \Bigl(1+\frac{A}{k}\Bigr) \leq \Bigl(\frac{x}{\lambda}\Bigr)^A.
\]
To check this condition for given mean $m = \lambda \Gamma\bigl(1+\frac{1}{k}\bigr)$ and variance $s^2 = \lambda^2 \bigl(\Gamma\bigl(1+\frac{2}{k}\bigr) - \Gamma\bigl(1+\frac{1}{k}\bigr)^2\bigr)$ of the Weibull distribution, we first have to solve numerically
\[
m^2 \Gamma\Bigl(1+\frac{2}{k}\Bigr) = (m^2 +s^2)\Gamma\Bigl(1+\frac{1}{k}\Bigr)^2
\]
for $k$ (note that there is no uniqueness issue as $(0, \infty) \ni z \mapsto \Gamma(1+2z)/\Gamma(1+z)^2$ is strictly increasing from $1$ to $\infty$), and then calculate $\lambda = m/\Gamma\bigl(1+\frac{1}{k}\bigr)$. One checks numerically that indeed all expected means are attainable and one can get an explicit expression for the expected selling time,
\[
\E[\tau] =  \frac{2\bigl(\log(\frac{x}{\lambda})+\frac{\gamma}{k}\bigr)}{\sigma^2 - 2(\mu-r)},
\]
where $\gamma$ is the Euler--Mascheroni constant.

The sufficient condition for a bounded embedding before time $T$ translates here into
	\[
	-\frac{\varphi(z)}{\bigl(1-\Phi(z)\bigr)\log\bigl(1-\Phi(z)\bigr)} \leq k \sigma \sqrt{T} 
	\]
	which is not satisfied for any finite $T$ as the left-hand expression diverges for $z \to -\infty$.

\textbf{Gamma}: In the case of a gamma distribution, $F_{\alpha,\beta} \sim \Gamma(\alpha,\beta)$, $F_{\alpha,\beta}(z) = \frac{\beta^\alpha}{\Gamma(\alpha)}\int_0^z  y^{\alpha -1}e^{-\beta y}\, \ud y$,
	\[
	\int_0^\infty z^A \, dF_{\alpha,\beta}(z) = \frac{\beta^\alpha}{\Gamma(\alpha)}\int_0^\infty z^A z^{\alpha -1} e^{-\beta z} \, \ud z = \frac{\Gamma(\alpha + A )}{\Gamma(\alpha)\beta^{A}}
	\]
	and given the initial asset value $x$, the attainability condition \eqref{moment} reduces to
	\[
	\frac{\Gamma(\alpha + A)}{\Gamma(\alpha)\beta^{A}} \leq x^A.
	\]
	In terms of mean $m = \frac{\alpha}{\beta}$ and variance $s^2 = \frac{\alpha}{\beta^2}$ this reads
	\[
	\Gamma\biggl(\frac{m^2}{s^2} + A \biggr) \leq \Bigl(\frac{x}{m}\Bigr)^A\Bigl(\frac{m^2}{s^2}\Bigr)^{A}\Gamma\biggl(\frac{m^2}{s^2}\biggr).
	\]
	Note that for $A \in(0,1)$ the function $(0, \infty) \ni z \mapsto \Gamma(A+z)/(z^A\Gamma(z))$ is strictly increasing from $0$ to $1$. Thus distributions with arbitrarily large expected mean are attainable. Also here the expected selling time is available in closed form,
\[
\E[\tau] =  \frac{2\bigl(\log(\beta x) - \psi(\alpha)\bigr)}{\sigma^2 - 2(\mu-r)},
\]
where $\psi = \frac{\Gamma'}{\Gamma}$ is the digamma function (polygamma of order zero).

	\begin{figure}
	\includegraphics[width=0.31\linewidth]{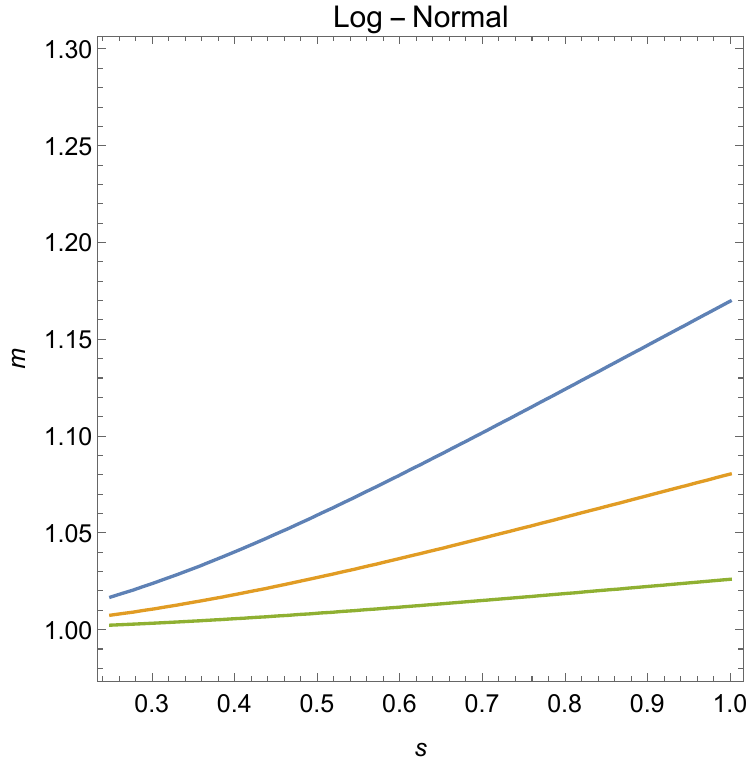}\quad
	\includegraphics[width=0.31\linewidth]{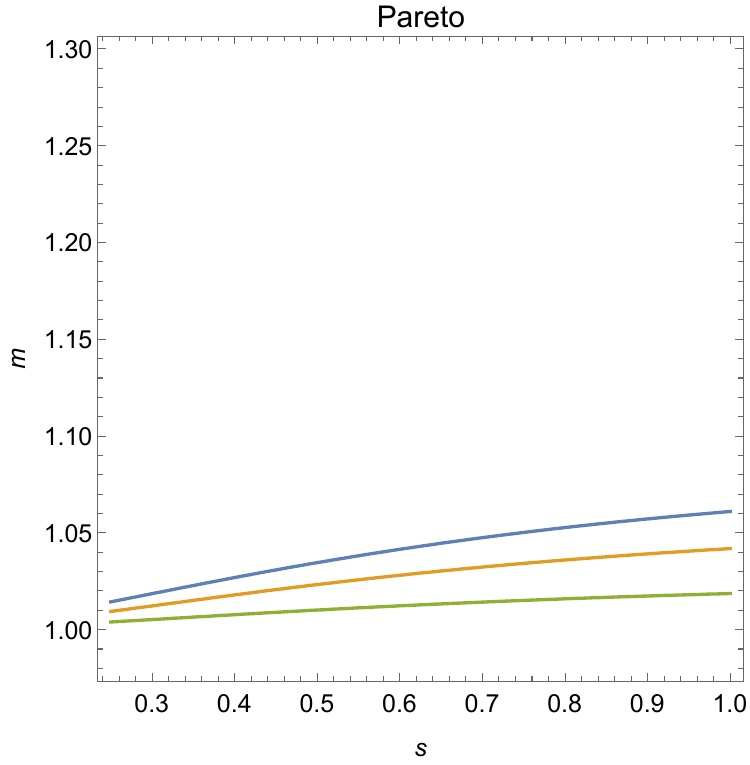}\quad
	\includegraphics[width=0.31\linewidth]{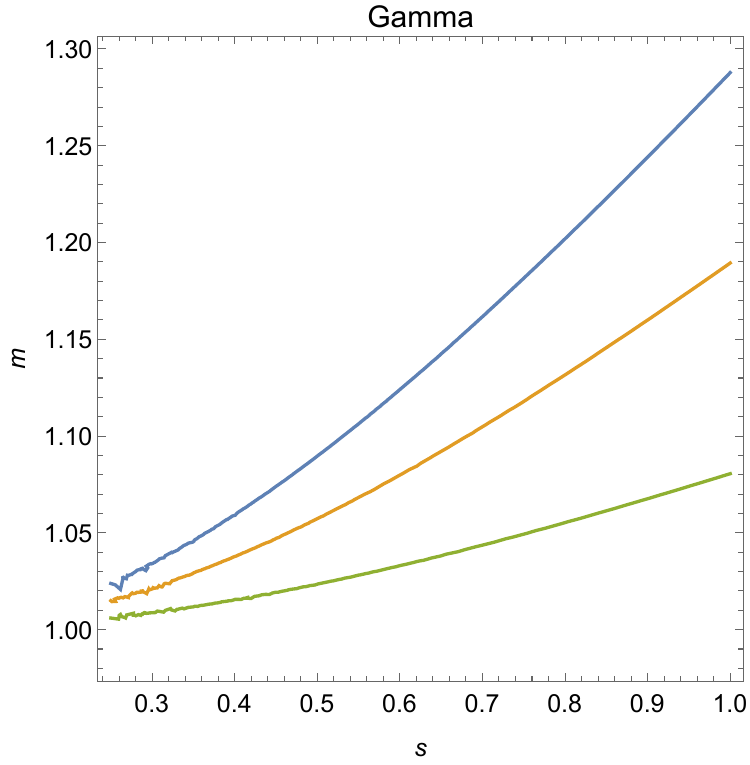}
		\caption{Attainable frontiers (mean-standard deviation trade-off) for different distributions and $A = 0.2$ (blue), $A = 0.5$ (orange), and $A = 0.8$ (green) (with $x =1$ throughout): log-normal, Pareto, and gamma distributions.}
	\end{figure}

\subsection{Comparison with Related Results}

The literature on the optimal timing of an asset sale broadly comprises two different streams. On the one hand, there is the idea of having a universal criterion which everybody can agree on, similar to the Kelly criterion in classical portfolio optimization. The idea here is to chase the running maximum on a finite time horizon. While the maximum occurs at a random time that is not a stopping time, one tries to approximate it by a stopping time. The criterion proposed in Shiryaev, Xu and Zhou \cite{SXZ08} is to minimize the relative error, i.e., $\inf_{0\leq \tau\leq T} \E\bigl[\frac{M_T - X_\tau}{M_T}\bigr]$, for a geometric Brownian motion $X$ with maximum process $M$, $M_t = \sup_{0 \leq s \leq t} X_s$. Du Toit and Peskir \cite{dTP09} study the equivalent problem $\sup_{0\leq \tau\leq T} \E\bigl[\frac{X_\tau}{M_T}\bigr]$ as well as the complementary $\inf_{0\leq \tau\leq T} \E\bigl[\frac{M_T}{X_\tau}\bigr]$. They find that for the former problem the optimal solution is to sell immediately if $\mu < \sigma^2/2$ and to hold until the end if $\mu > \sigma^2/2$. For the latter problem, however, besides selling immediately for low $\mu$ and holding for high $\mu$, there is an intermediate region of $0 < \mu < \sigma^2$ in which the optimal stopping time is a time where $M_t/X_t$ hits a barrier. This barrier is described qualitatively as a finite, decreasing function.

Literature that takes the seller's preferences into account includes Henderson \cite{Hen12}, who considers the question of optimal timing of a sales operation in a prospect-theoretical context. Similar to our general model she works with diffusions, and points out that the main behavioral difference depends on whether the domain of the scale function extends on the left to $-\infty$ or not. The optimal strategy depends also on the parameter $A$, however the solution is more complicated and depends on the utility function as well. Pedersen and Peskir \cite{PP16} take a mean-variance approach to the optimal asset sale in a geometric Brownian motion model. Similar to our paper, they find that if $\mu \geq \sigma^2/2$ it is best to never sell, if $\mu \leq 0$ to sell immediately and in between the static optimum is to sell when the asset hits a constant barrier depending on the variable of the mean-variance trade-off. Leung and Wang \cite{LW19} consider the problem of optimally selling in a geometric Brownian motion (as well as exponential Ornstein--Uhlenbeck) model according to expected (discounted) utility criteria. Besides the cases of immediate sale or infinite holding, the optimal selling occurs at the hitting time of a constant barrier (depending on the utility functions). Most importantly, in the case of preferences with power utility they conclude that depending on the risk-aversion parameter, either immediate selling or indefinite holding are optimal. We note in particular that in the last two works discussed, the barrier hitting times are not a.s. finite if the barrier is above the initial asset level. Under the given parameter configuration, the geometric Brownian motion converges to $0$ at infinity, only a zero contribution to the utility/mean-variance criterion is recorded. This stands in stark contrast to our approach where we require stopping times to be a.s. finite. In the intermediate geometric Brownian motion case in \cite{PP16} as well as in \cite{LW19} when $\mu<\sigma^2/2$, an upper threshold may fail to be reached and the limiting contribution on non-hitting is taken as zero. Those defective stopping times contrast with our requirement that stopping times be almost surely finite.

\section{General Theory -- Diffusion Processes}\label{sec:gen}

We assume that the (subjectively) discounted asset process $X$ follows a (potentially explosive) diffusion process on the interval $I =(\ell,\varrho)$, $-\infty \leq \ell <\varrho \leq \infty$,
\begin{equation}\label{diffusion}
\ud X_t = \mu(X_t) \, \ud t + \sigma(X_t) \, \ud W_t, \qquad X_0 =x, 
\end{equation}
up to $e = \inf\{ t >0 \, : \, X_t \notin (\ell, \varrho)\}$ where we assume that $\mu, \sigma$ are Borel-measurable, $\sigma >0$ and $\frac{\mu}{\sigma^2}$, $\frac{1}{\sigma^2}$ are locally integrable everywhere on $I$. This guarantees the existence of a weak solution to the SDE \eqref{diffusion} up to explosion (see \cite[Section 5.5]{KS91}).  We assume that the boundaries are absorbing, i.e., after the explosion we have $X_t = X_e$, $t > e$.  
\begin{definition}
	The \emph{scale function} of the diffusion process $X$ is given as
		\[
		\varsigma(z) := \int_x^z e^{-2\int_x^y \frac{\mu(u)}{\sigma^2(u)}\, \ud u}\, \ud y
		\]
		for $x = X_0$. The \emph{scaled mean} of a distribution is defined as  $m_F^\varsigma = \int \varsigma(z) \, F(\ud z)$ provided the integral converges. We write just $m_F$ for the unscaled mean (i.e., $\varsigma(z) = z$ for all $z \in I$).
\end{definition}

From the general theory of the Skorokhod embedding theorem, we have the following:

\begin{theorem}\label{thm:sko-gen}
	Let $F$ be a distribution with support within $(\ell, \varrho)$. $F$ is attainable if and only if one of the following cases is true:
	\begin{itemize}
		\item[i)] $-\infty = \varsigma(\ell) < \varsigma(\varrho) = \infty$;
		\item[ii)] $-\infty < \varsigma(\ell) < \varsigma(\varrho) = \infty$, $m_F^\varsigma$ exists and $m_F^\varsigma \leq \varsigma(x)=0$;
		\item[iii)] $-\infty = \varsigma(\ell) < \varsigma(\varrho) < \infty$, $m_F^\varsigma$ exists and $m_F^\varsigma \geq \varsigma(x)=0$;
		\item[iv)] $-\infty < \varsigma(\ell) < \varsigma(\varrho) < \infty$, $m_F^\varsigma$ exists and $m_F^\varsigma = \varsigma(x)=0$.
	\end{itemize} 
\end{theorem}

\begin{proof}
	This follows from the result of Pedersen and Peskir \cite[Theorem 2.1]{PP01} with a slight extension by Cox and Hobson given in \cite[Lemma 9]{CH04}.
\end{proof}

We can derive the following statement on super-attainability of distributions:
\begin{theorem}
	Let $F$ be a distribution supported on $(\ell, \varrho)$.
	\begin{itemize}
		\item[i)] If $-\infty < \varsigma(\ell)$, $F$ is super-attainable if and only if $m_F^\varsigma$ exists and $m_F^\varsigma \leq \varsigma(x) = 0$.
		\item[ii)] If $-\infty = \varsigma(\ell)$, $F$ is always super-attainable.
	\end{itemize}  
\end{theorem}

\begin{proof}
	The proof is identical to that of Theorem \ref{super-gbm} in the geometric Brownian motion case.
\end{proof}

In the case that all distributions are attainable, there is no need to plan for a sale at a given distribution. On the other hand, in the case $m_F^\varsigma \leq \varsigma(x)$, if the scale function is convex an immediate sale is advised due to uncertainty aversion: all attainable distributions have mean not above the initial asset value. Indeed, for an arbitrary distribution $F$
\[
\varsigma(x) \geq \int \varsigma(z) \, F(\ud z) \geq  \varsigma\biggl( \int z \, F(\ud z) \biggr) = \varsigma(m_F)
\]
by Jensen's inequality, and thus, as $\varsigma$ is increasing by definition, $x \geq m_F$. A direct calculation shows that this is equivalent to $\mu(z) \leq 0$ for all $z \in I$.

Conversely, if $\varsigma$ is strictly concave (equivalent to $\mu(z) > 0$ for all $z \in I$) and $m_F^\varsigma = \varsigma(x)$, then we have that the mean of $F$ is always higher than the current asset price:
\[
\varsigma(x) = \int \varsigma(z) \, F(\ud z) <  \varsigma\biggl( \int z \, F(\ud z) \biggr) = \varsigma(m_F),
\]
and as $\varsigma$ is strictly increasing the result follows. 

We can give a more precise argument when a diffusion admits an attainable distribution with mean above the asset value. The description relies on a notion of \emph{local-global convexity/concavity} about the current asset value $x$, i.e., on the comparison of the slope of the tangent of the scale function at $x$ to the slope of all secants drawn from $x$.

\begin{proposition}\label{precise-mean}
	Let $X$ be a diffusion process with $X_0=x$ and scale function $\varsigma$ satisfying $\varsigma(\ell) > -\infty$.
	\begin{itemize}
		\item[a)] There is no distribution with finite mean that the mean is above the current asset value if for all $\underline{x}, \overline{x} \in (\ell, \varrho)$ with $\underline{x} < x < \overline{x}$ we have
		\begin{equation}\label{loc-glob-conv}
		\frac{\varsigma(x) - \varsigma(\underline{x})}{x -\underline{x}} \leq \varsigma'(x) \leq \frac{\varsigma(\overline{x})-\varsigma(x)}{\overline{x} - x} .
		\end{equation}
		\item[b)] Every optimal distribution $F$ with finite mean has mean strictly above current asset value if there exist $\underline{x}, \overline{x} \in (\ell, \varrho)$ with $\underline{x} < x < \overline{x}$ such that
		\begin{equation}\label{ex-F-ab}
		\frac{\varsigma(x) - \varsigma(\underline{x})}{x -\underline{x}} > \frac{\varsigma(\overline{x})-\varsigma(x)}{\overline{x} - x} .
		\end{equation}
		\item[c)] Every optimal distribution $F$ with finite mean has mean strictly below the current asset value if there exist $\underline{x}, \overline{x} \in (\ell, \varrho)$ with $\underline{x} < x < \overline{x}$ such that
		\begin{equation}\label{ex-F-bel}
        \frac{\varsigma(x) - \varsigma(\underline{x})}{x -\underline{x}} < \frac{\varsigma(\overline{x})-\varsigma(x)}{\overline{x} - x} .		\end{equation}
		\item[d)] All optimal distributions with finite mean have mean of at least the current asset value if for all $\underline{x}, \overline{x} \in (\ell, \varrho)$ with $\underline{x} < x < \overline{x}$ we have
		\begin{equation}\label{loc-glob-conc}
		\frac{\varsigma(x) - \varsigma(\underline{x})}{x -\underline{x}} \geq \varsigma'(x) \geq \frac{\varsigma(\overline{x})-\varsigma(x)}{\overline{x} - x}.
		\end{equation}
		For a non-constant distribution the mean is strictly above the current asset value if one of the two inequalities holds strictly on a set of positive $F$-mass.
		\end{itemize}
\end{proposition}

\begin{proof}
	\begin{itemize}
		\item[a)] By an argument as in the proof of Theorem \ref{super-gbm} we can restrict ourselves to the case of optimal distributions. Assume \eqref{loc-glob-conv} holds. Then, 
		\begin{align*}
		0 &= \int \varsigma(z) \, F(\ud z) \geq \int \Bigl(\varsigma(x) + \varsigma'(x)(z-x)\Bigr) \, F(\ud z) = \varsigma'(x) \biggl(\int z \, F(\ud z) - \int x \, F(\ud z)\biggr)
		\\ &= \varsigma'(x) \bigl(m_F - x)
		\end{align*}
		and thus $m_F \leq x$ as the scale function is increasing.
		\item[b)] Define now the distribution function
		\[
		F (z) = \left\{ \begin{array}{ll} 
		0 & \text{ if } z < \underline{x},\\
		\frac{\varsigma(\overline{x}) -\varsigma(x)}{\varsigma(\overline{x})-\varsigma{(\underline{x})}} & \text{ if } \underline{x} \leq z < \overline{x},\\
		1 & \text{ if } z \geq \overline{x}.
		\end{array} \right.
		\]
		We have then
		\[
		\int \varsigma(z) \, F(\ud z)  = \frac{\varsigma(\overline{x})-\varsigma(x)}{\varsigma(\overline{x})-\varsigma{(\underline{x})}}\varsigma(\underline{x}) + \frac{\varsigma(x)-\varsigma(\underline{x})}{\varsigma(\overline{x})-\varsigma{(\underline{x})}} \varsigma(\overline{x}) = \varsigma(x)
		\]
		thus $F$ is optimal. Moreover,
		\begin{align*}
		\int z \, F(\ud z) & = \frac{\varsigma(\overline{x})-\varsigma(x)}{\varsigma(\overline{x})-\varsigma{(\underline{x})}}\underline{x} + \frac{\varsigma(x)-\varsigma(\underline{x})}{\varsigma(\overline{x})-\varsigma{(\underline{x})}} \overline{x} >x
		\end{align*}
		proving the statement.
		\item[c)] This works exactly as in b), with the inequality reversed. 
		\item[d)] The proof is the converse of a). As the scale function is strictly increasing and every non-constant distribution has to take values on both sides of the mean with positive probability, the inequality is strict if at least one of the inequalities in the condition is strict on a set in the support of $F$.
	\end{itemize}
\end{proof}

A different characterization when immediate sale is (not) advisable can be given in terms of a space-time average of the drift of the asset price dynamics. 
\begin{proposition}
	Assume that the local martingale part $M$ of the Doob--Meyer decomposition $X = M + A$ of $X$ is a martingale of class (D). Then there is an attainable distribution $F$ with mean $m_F$ above the current asset value $x$ if there exists a stopping time $\tau$ with $\mathbb{E}\bigl[\int_0^\tau \mu\bigl(X_s\bigr) \, \ud s \bigr] > 0.$ Vice versa, if for all stopping times $\tau$ we have $\mathbb{E}\bigl[\int_0^\tau \mu\bigl(X_s\bigr) \, \ud s \bigr] \leq 0$, the mean of all attainable distributions is at or below the current asset price.
\end{proposition}

\begin{proof}
	We note that if such $\tau$ exists, we have for $F$ such that $X_\tau  \sim F$
	\[
	m_F = \int z \, F(\ud z) = \mathbb{E}\bigl[ X_\tau\bigr] = \mathbb{E}\biggl[ x + \int_0^\tau \mu\bigl(X_s\bigr) \, \ud s + \int_0^\tau \sigma\bigl(X_s\bigr) \, \ud W_s\biggr] = x + \mathbb{E}\biggl[\int_0^\tau \mu\bigl(X_s\bigr) \, \ud s\biggr] > x.
	\]
	The other way round, if $F$ is attainable with $m_F>x$, by definition there exists a $\tau$ with $X_\tau \sim  F$ and 
	\[
\mathbb{E}\biggl[\int_0^\tau \mu\bigl(X_s\bigr) \, \ud s\biggr] = \mathbb{E}\bigl[ X_\tau\bigr] - \mathbb{E}\biggl[ x + \int_0^\tau \sigma\bigl(X_s\bigr) \, \ud W_s\biggr] = \int z \, F(\ud z) - x = m_F - x>0.
\]
The second statement follows analogously.
\end{proof}

From this we can derive an easier-to-check sufficient condition when immediate sale is suboptimal or optimal:

\begin{corollary}
	If $\mathbb{E}\bigl[\int_0^\tau \sigma(X_s)^2\, \ud s \bigr] < \infty$ and $\mathbb{E}\bigl[\int_0^\tau \mu(X_s)\, \ud s \bigr] \geq 0$ for all stopping times $\tau<\infty$ a.s. an there exists an a.s. finite stopping time $\tau$ with $\mathbb{E}\bigl[\int_0^\tau \mu(X_s)\, \ud s \bigr] > 0$, then $F \sim X_\tau$ is attainable and $m_F > x$. Vice versa, if $\mathbb{E}\bigl[\int_0^\tau \sigma(X_s)^2\, \ud s \bigr] < \infty$ and $\mathbb{E}\bigl[\int_0^\tau \mu(X_s)\, \ud s \bigr] \leq 0$ for all $\tau <\infty$ a.s. we have $m_F \leq x$.
\end{corollary}

We also want to provide an explicit stopping rule that attains a targeted optimal sale price distribution. This can be done by a generalization of the Az\'{e}ma--Yor embedding \cite{AY79a, AY79b} due to Pedersen and Peskir \cite{PP01}: For $m_F^\varsigma \geq \varsigma(x)=0$, define the extended scaled barycenter function $\Psi_F^{\varsigma,+}$ of the right tail of $F$ (supported on $(\ell, \varrho)$) by setting $\Psi_F^{\varsigma,+} (x)=-\infty$ when $x \leq \varsigma^{-1}(m_F^\varsigma)$, $\Psi_F^{\varsigma,+} (x) =x$ for $x \geq \varsigma^{-1}(\beta)$ with $\beta := \sup\bigl\{y \in \R \, : \, F\bigl(\varsigma^{-1}(y-)\bigr) < 1\bigr\}$ and
\[
\bigl(\Psi_F^{\varsigma,+}\bigr)^{-1} (z) = \varsigma^{-1}\Biggl( \frac{1}{1-F(z-)}\int_{[z,\varrho)} \varsigma(u) \, F( \ud u) \Biggr)
\]
for $\varsigma^{-1}(m_F^\varsigma) < z < \varsigma^{-1}(\beta)$, where, as usual, $F(z-) := \lim_{\varepsilon \downarrow 0} F(z-\varepsilon)$. Then we define the stopping time $\tau^*$ as the first hitting time
\begin{align*}
\tau^* & =\inf\bigl\{ t > 0 \, : \, \Psi_F^{\varsigma,+}(M_t) \geq X_t\bigr\}\\
& =\inf\bigl\{ t > 0 \, : \, (M_t, X_t) \in \mathcal{D}_F \bigr\}, \quad \mathcal{D}_F := \bigl\{(m,x) \in (\ell,\varrho)^2 \, : \, \Psi_F^{\varsigma,+}(m) \geq x\bigr\}
\end{align*}
where $M_t = \sup_{s \leq t} X_s$.
In the case that $m_F^\varsigma \leq \varsigma(x)=0$, the corresponding object is the extended scaled barycenter function of the left tail of $F$ (supported on $(\ell, \varrho)$) defined by$\Psi_F^{\varsigma,-} (x)=+\infty$ when $x \geq \varsigma^{-1}(m_F^\varsigma)$, $\Psi_F^{\varsigma,-} (x) =x$ for $x \leq \varsigma^{-1}(\alpha)$ with $\alpha := \inf\bigl\{y \in \R \, : \, F\bigl(\varsigma^{-1}(y)\bigr) > 0\bigr\}$ and
\[
\bigl(\Psi_F^{\varsigma,-}\bigr)^{-1} (z) = \varsigma^{-1}\Biggl( \frac{1}{F(z)}\int_{(\ell,z]} \varsigma(u) \, F( \ud u) \Biggr)
\] 
for $\varsigma^{-1}(\alpha) < z < \varsigma^{-1}(m_F^\varsigma)$. $\tau^*$ is the first hitting time
\begin{align*}
\tau^* & =\inf\bigl\{ t > 0 \, : \, \Psi_F^{\varsigma,-}(m_t) \leq X_t\bigr\}\\
& =\inf\bigl\{ t > 0 \, : \, (m_t, X_t) \in \mathcal{D}_F \bigr\}, \quad \mathcal{D}_F := \bigl\{(m,x) \in (\ell,\varrho)^2 \, : \, \Psi_F^{\varsigma,-}(m) \leq x \bigr\}
\end{align*}
where $m_t = \inf_{s \leq t} X_s$.

This embedding has the nice additional property that it maximizes the maximum (resp. minimizes the minimum) over all stopping times satisfying $X_\tau \sim F$ stochastically (in first-order dominance),
\[
\P\bigl[\sup_{s \leq \tau^*} X_s \geq z \bigr] \geq \P\bigl[\sup_{s \leq \tau} X_s \geq z\bigr],
\]
in a suitable class, e.g., the stopping times with $\mathbb{E}\bigl[\sup_{0\leq s \leq \tau} \vert \varsigma(X_s)\vert\bigr]< \infty$ (Pedersen and Peskir \cite[Proposition 3.1]{PP01}) or the class of minimal stopping times (Cox and Hobson \cite[Corollary 17]{CH06}). 

Finally, while no closed form for the expected stopping time is available, we can give for the case $m_F^\varsigma =0$ the following description via the solution of an ordinary differential equation.

\begin{proposition}
	Let $g$ be a solution to the initial value problem
	\begin{align}\label{Dynkin-ODE}
	\left\{ \begin{array}{rl} \mu(z) g'(z) + \frac{\sigma^2(z)}{2}g''(z) & =1,\\
	g(x) & = 0.\end{array}\right.
	\end{align}
	Then, if finite,  the expected stopping time $\tau$ for the sale is given by
	\[
	\E[\tau] = \int g(z) \, F(\ud z).
	\]
\end{proposition}

\begin{proof}
	Let
	\[
	L = \mu(z) \frac{\ud}{\ud z} + \frac{\sigma^2(z)}{2}\frac{\ud^2}{\ud z^2},
	\]
	then $L$ is the generator of the Markov process $X$, and for all functions $f$ in its domain we have Dynkin's formula (see \cite[Section 7.4]{Oks03})
	\[
	\E\bigl[f\bigl(X_\tau\bigr)\bigr] - f(x) = \E \biggl[ \int_0^\tau \bigl(Lf\bigr)(X_s) \, \ud s\biggr].
	\]
	Now as $g$ is the solution to the ODE \eqref{Dynkin-ODE}, the integrand is just 1 and thus
	\[
\E[\tau] = \E\bigl[g\bigl(X_\tau\bigr)\bigr] = \int g(z) \, F(\ud z).
\]
\end{proof}

\begin{remark}
	It might look curious that the ODE \eqref{Dynkin-ODE} seems under-determined (a second-order ODE with a single initial condition). But one checks that for the case $ m_F^\varsigma =0$, the scale function $\varsigma$ is a solution to the \emph{homogeneous} ODE $\mu(z) g'(z) + \sigma^2(z)/2g''(z) =0$, $g(x)=0$, and thus once a solution $g$ is found, all solutions are given by $g(z) + \lambda \varsigma(z)$, $\lambda \in \mathbb{R}$ as $\int \varsigma(z) \, F(\ud z)=0$ and it does not matter which of the solutions we choose.
\end{remark}

To answer the question of which distributions can be achieved in a given bounded time horizon, we can rely on the answer given in \cite[Theorem 6]{AHS15}.

\begin{proposition}
	Assume that the volatility of the diffusion is positive, differentiable and the function
	\[
	z \mapsto \sigma'(z) -\frac{2\mu(z)}{\sigma(z)} 
	\]
	is non-increasing. If the distribution $F$ satisfies $\int \varsigma(z) F(\ud z) = 0$ and $g := F^{-1} \circ \Phi$ is absolutely continuous, and 
	\[
	\frac{g'(z)}{\sigma\bigl(g(z)\bigr)} \leq \sqrt{T}.
	\]
    then there exists a stopping time $\tau 
\leq T$ such that $X_\tau \sim F$.
\end{proposition}

\section{Conclusion}\label{sec:conc}

The article presented a novel method to determine the optimal time to sell an asset, taking into  account the seller's individual preferences using the distribution builder methodology introduced by Sharpe, Goldstein and Blythe \cite{SGB00}. We treat the problem first for geometric Brownian motion (the model used in most of the literature on optimal timing) as well as on general one-dimensional diffusion models, a close connection to the Skorokhod problem of embedding distributions into diffusion processes via optimal stopping is uncovered. The main extension is that we are not only interested in which distributions can be achieved by stopping the diffusion process, but also those that can be dominated (in the sense of first order stochastic dominance).

We provide an explicit strategy for selling the asset based on the Az\'{e}ma--Yor construction of the Skorokhod embedding, characterize in which cases it is optimal to sell immediately or hold the asset indefinitely, and determine the expected time of liquidation. In the case of a geometric Brownian motion and two-parameter families of distributions, we show that the problem gives rise to a mean--variance trade-off, that higher expected returns can be achieved at the price of higher variance, and show how this can take different forms depending on specific target distributions.

\bibliographystyle{alpha}
\bibliography{PC_bib}
 
 \appendix

\section{Further Examples}

We give several further explicit examples besides the classical geometric Brownian motion discussed at length in Section \ref{sec:GBM}.  

\subsection{Drifted Brownian Motion}

We consider the cases where the discounted asset price process $X$ follows a drifted Brownian motion 
\[
\ud X_t = \mu \, \ud t + \sigma \, \ud W_t, \qquad X_0 = x.
\]
In this case, it follows directly from \cite{GF00} that a distribution $F$ is attainable if and only if $\int e^{-\frac{2\mu (z-x)}{\sigma^2}} \, F(\ud z) \leq 1$. By the same arguments as in the geometric Brownian motion case, if $\mu \geq 0$, all distributions are super-attainable, and if $\mu < 0$ only those with $\int e^{-\frac{2\mu (z-x)}{\sigma^2}} \, F(\ud z) \leq 1$. However, as the latter implies that the mean of $F$ has to be smaller than $x$, quite distinctly here the model splits just into two cases: there is no reason to sell if $\mu \geq 0$ and immediate sale is suggested if $\mu<0$.

\subsection{(Exponential) Ornstein--Uhlenbeck Process}

We consider the cases where the discounted asset price process $X$ follows an Ornstein--Uhlenbeck process 
\[
\ud X_t = \kappa (\lambda - X_t) \, \ud t + \sigma \, \ud W_t, \qquad X_0 = x,
\]
with $x, \lambda \in \R$, $\kappa, \sigma >0$. In this case the scale function is
\[
\varsigma(z) = e^{-\frac{\kappa(x-\lambda)^2}{\sigma^2}} \int_x^z e^{\frac{\kappa(y-\lambda)^2}{\sigma^2}} \, \ud y
\]
and thus $\varsigma(\pm \infty ) = \pm \infty$. It follows that all distributions are attainable and a fortiori super-attainable. The same is true for the exponential Ornstein--Uhlenbeck process often used in financial modeling: by an exponential transform one sees that all distributions with positive support are attainable and therefore all distributions (on the real line) are super-attainable. In both models one would never sell the asset.

\subsection{CIR Process}

In the case where $X$ follows a CIR process on $[0,\infty)$ for the discounted asset price
\[
\ud X_t = \kappa (\lambda - X_t) \, \ud t + \sigma \sqrt{X_t}\, \ud W_t, \qquad X_0 = x,
\]
with parameters $x, \kappa, \lambda, \sigma > 0$, the scale function is
\[
\varsigma(z) =  \int_x^z  \Bigl(\frac{y}{x}\Bigr)^{-\frac{2\kappa \lambda}{\sigma^2}} e^{\frac{2\kappa}{\sigma^2}(y-x)}\, \ud y 
\]
and thus $\varsigma(0) = - \infty$ if $2\kappa \lambda \geq \sigma^2$ and $\varsigma(0) > - \infty$ if $2\kappa \lambda < \sigma^2$ while $\varsigma(\infty) =\infty$ always. It follows that in the case that $\kappa >0$ and the Feller condition $2\kappa \lambda \geq \sigma^2$ holds all distributions are attainable. However, if the Feller condition does not hold, the process reaches $0$ in finite time with positive probability and by our conventions is absorbed at the boundary; only distributions that satisfy 
\[
\int_0^\infty \int_x^z \Bigl(\frac{x}{y}\Bigr)^{\frac{2\kappa \lambda}{\sigma^2}} e^{\frac{2\kappa}{\sigma^2}(y-x)} \, \ud y \, F(\ud z) \leq 0
\]
can be attained. This can be simplified to
\[
\int_0^\infty y^{-\frac{2\kappa \lambda}{\sigma^2}} e^{\frac{2\kappa}{\sigma^2}y}\bigl(F(y-) - F_{\delta_x}(y)\bigr) \, \ud y \geq 0.
\]

\end{document}